\documentclass[graybox,envcountsame,envcountsect]{svmult}

\usepackage{mathptmx}       
\usepackage{helvet}         
\usepackage{courier}        
\usepackage{type1cm}        
\usepackage{makeidx}         
\usepackage{graphicx}        
\usepackage{multicol}        
\usepackage[bottom]{footmisc}

\usepackage{amsfonts}
\usepackage{amsmath}
\usepackage{amssymb}
\usepackage{graphicx}

\usepackage{verbatim}
\usepackage{float}

\usepackage{enumerate}

\usepackage{xcolor}
\usepackage{subfigure}

\usepackage{latexsym,ifthen}

\usepackage{listings}
\newcommand{\bsigk}{\boldsymbol{\Sigma}_{\mathbf k}^{\mathbf 0}}
\newcommand{\bsig}[1]{\boldsymbol{\Sigma}_{{\mathbf #1}}^{\mathbf
0}}
\newcommand{\bpi}[1]{\boldsymbol{\Pi}_{{\mathbf #1}}^{\mathbf
0}}
\newcommand{\bpik}{\boldsymbol{\Pi}_{\mathbf k}^{\mathbf 0}}

\newcommand{\loc}{\mathrm{loc}}

\DeclareMathOperator{\dimloc}{\dim_{\loc}}

\DeclareMathOperator{\Dimloc}{\Dim_{\loc}}

\newcommand{\B}{\mathrm{B}}
\newcommand{\miI}{\mathrm{I}}
\newcommand{\spec}{\mathrm{sp}}

\newcommand{\sdim}{\mathrm{sdim}}
\newcommand{\sinf}{\Sigma^{\infty}}

\newcommand{\range}{\mathrm{range}}
\newcommand{\myset}[2]{ \left\{ #1 \left| \, #2 \right. \right\} }

\newcommand{\ignore}[1]{}

\newcommand{\N}{\mathbb{N}}
\newcommand{\Z}{\mathbb{Z}}
\newcommand{\Q}{\mathbb{Q}}
\newcommand{\R}{\mathbb{R}}
\newcommand{\Rn}{\R^n}

\newcommand{\dimh}{\mathrm{dim}_\mathrm{H}}

\newcommand{\cdim}{\mathrm{cdim}}
\newcommand{\adim}{\mathrm{cdim}}

\renewcommand{\dim}{{\mathrm{dim}}}

\newcommand{\dimpack}{{\dim}_{\mathrm{P}}}

\newcommand{\mdim}{\mathrm{mdim}}
\newcommand{\Mdim}{\mathrm{Mdim}}

\newcommand{\Dim}{{\mathrm{Dim}}}
\newcommand{\cDim}{{\mathrm{cDim}}}
\newcommand{\aDim}{{\mathrm{cDim}}}

\newcommand{\str}{{\mathrm{str}}}

\newcommand{\strings}{\{0,1\}^*}

\newcommand{\DIM}{\mathrm{DIM}}

\newcommand{\K}{{\mathrm{K}}}

\newcommand{\dom}{\mathrm{dom}}

\newcommand{\length}{{\mathrm{length}}}

\newenvironment{example*}[1]
{ {\noindent {\bf Example #1.}} } {  }
\newenvironment{claim*}[1]
{ {\noindent {\bf Claim #1.}} } {  }
\newenvironment{theorem*}
{ {\noindent {\bf Theorem.}} } {  }

\numberwithin{equation}{section}

\makeindex             

\begin{document}

\title*{Algorithmic Fractal Dimensions in Geometric Measure Theory}
\author{Jack H. Lutz and Elvira
Mayordomo}
\institute{Jack H. Lutz \at Department of Computer Science, Iowa
State University, Ames, IA 50011 USA, \email{lutz@cs.iastate.edu}
\and Elvira Mayordomo \at Departamento de Inform\'atica e
Ingenier\'ia de Sistemas, Instituto de Investigaci\'on en
Ingenier\'{\i}a de Arag\'on, Universidad de Zaragoza, 50018
Zaragoza, SPAIN \email{elvira@unizar.es}}
%
%
\maketitle

\abstract*{The development of algorithmic fractal dimensions in this
century has had many fruitful interactions with geometric measure
theory, especially fractal geometry in Euclidean spaces.  We survey
these developments, with emphasis on connections with computable
functions on the reals, recent uses of algorithmic dimensions in
proving new theorems in classical (non-algorithmic) fractal
geometry, and directions for future research.}

\abstract{The development of algorithmic fractal dimensions in this
century has had many fruitful interactions with geometric measure
theory, especially fractal geometry in Euclidean spaces.  We survey
these developments, with emphasis on connections with computable
functions on the reals, recent uses of algorithmic dimensions in
proving new theorems in classical (non-algorithmic) fractal
geometry, and directions for future research.}

\section{Introduction}\label{sec:1}
In early 2000, classical Hausdorff dimension \cite{Haus19}\ was
shown to admit a new characterization in terms of betting strategies
called martingales \cite{DCCcon}.  This characterization enabled the
development of various effective, i.e., algorithmic, versions of
Hausdorff dimension obtained by imposing computability and
complexity constraints on these martingales.  These algorithmic
versions included resource-bounded dimensions, which impose
dimension structure on various complexity classes \cite{DCC}, the
(constructive) dimensions of infinite binary sequences, which
interact usefully with algorithmic information theory \cite{DISS},
and the finite-state dimensions of infinite binary sequences, which
interact usefully with data compression and Borel normality
\cite{FSD}.  Soon thereafter, classical packing dimension
\cite{Tric82, Sull84}\ was shown to admit a new characterization in
terms of martingales that is exactly dual to the martingale
characterization of Hausdorff dimension \cite{ESDAICC}.  This led
immediately to the development of strong resource-bounded
dimensions, strong (constructive) dimension, and strong finite-state
dimension \cite{ESDAICC}, which are all algorithmic versions of
packing dimension. In the years since these developments, hundreds
of research papers by many authors have deepened our understanding
of these algorithmic dimensions.

   Most work to date on effective dimensions has been carried out in the Cantor space, which consists of all infinite
   binary sequences.  This is natural, because effective dimensions speak to many issues that were already being
   investigated in the Cantor space.  However, the classical fractal dimensions from which these effective dimensions
   arose--Hausdorff dimension and packing dimension--are powerful quantitative tools of geometric measure theory that have
   been most useful in Euclidean spaces and other metric spaces that have far richer structures than the totally
   disconnected Cantor space.

   This chapter surveys research results to date on algorithmic fractal dimensions in geometric measure theory, especially
   fractal geometry in Euclidean spaces.  This is a small fraction of the existing body of work on algorithmic fractal
   dimensions, but it is substantial, and it includes some exciting new results.

   It is natural to identify a real number with its binary expansion
   and to use this identification to define algorithmic dimensions
   in Euclidean spaces in terms of their counterparts in Cantor
   space. This approach works for some purposes, but it becomes a
   dead end when algorithmic dimensions are used in geometric
   measure theory and computable analysis. The difficulty, first
   noted by Turing in his famous correction \cite{Tur38}, is that
   many {\sl obviously\/} computable functions on the reals (e.g.,
   addition) are {\sl not\/} computable if reals are represented by
   their binary expansions \cite{Wei00}. We thus take a principled
   approach from the beginning, developing algorithmic dimensions in
   Euclidean spaces in terms of the quantity $\K_r(x)$ in the
   following paragraph, so that the theory can seamlessly advance to
   sophisticated applications.

   Algorithmic dimension and strong algorithmic  dimension are the most extensively investigated effective dimensions.  One
   major reason for this is that these algorithmic dimensions were shown by the second author and others \cite{May02,ESDAICC,
   LM08}\
    to have characterizations in terms of Kolmogorov complexity, the central notion of algorithmic
   information theory.  In Section \ref{sec:2}\ below we give a brief introduction to the Kolmogorov complexity $\K_r(x)$ of a point $x$ in
   Euclidean space at a given precision $r$.

    In Section \ref{sec:3}\ we use the above Kolmorogov complexity notion to develop the algorithmic dimension $\dim(x)$ and the
    strong algorithmic  dimension $\Dim(x)$ of each point $x$ in Euclidean space.  This development supports the useful intuition
    that these dimensions are asymptotic measures of the density of algorithmic information in the point $x$.  We discuss how
    these dimensions relate to the local dimensions that arise in the so-called thermodynamic formalism of fractal
    geometry; we discuss the history and terminology of algorithmic dimensions; we review the prima facie case that algorithmic dimensions are geometrically meaningful; and we discuss what
    is known about the circumstances in which algorithmic dimensions agree with their classical counterparts.  We then
    discuss the authors' use of algorithmic dimensions to analyze self-similar fractals \cite{LM08}.  This analysis gives
    us a new, information-theoretic proof of the classical formula of Moran \cite{Mora46}\ for the Hausdorff dimensions of
    self-similar fractals in terms of the contraction ratios of the iterated function systems that generate them.  This new
    proof gives a clear account of ``where the dimension comes from'' in the construction of such fractals.  Section \ref{sec:3}\
    concludes with a survey of the dimensions of points on lines in Euclidean spaces, a topic that has been surprisingly
    challenging until a very recent breakthrough by N. Lutz and Stull \cite{LutStu17}.

   We survey interactive aspects of algorithmic fractal dimensions in Euclidean spaces in Section \ref{sec:4}, starting with the
   mutual algorithmic dimensions developed by Case and the first author \cite{CasLut15}.  These dimensions, $\mdim(x:y)$ and
   $\Mdim(x:y)$, are analogous to the mutual information measures of Shannon information theory and algorithmic information
   theory.  Intuitively, $\mdim(x:y)$ and $\Mdim(x:y)$ are asymptotic measures of the density of the algorithmic information
   shared by points $x$ and $y$ in Euclidean spaces.  We survey the fundamental properties of these mutual dimensions, which
   are analogous to those of their information-theoretic analogs.  The most important of these properties are those that
   govern how mutual dimensions are affected by functions on Euclidean spaces that are computable in the sense of
   computable analysis \cite{Wei00}.  Specifically, we review the information processing inequalities of \cite{CasLut15},
   which state that $\mdim(f(x):y) \le \mdim(x:y)$ and $\Mdim(f(x):y) \le \Mdim(x:y)$ hold for all computable Lipschitz functions $f$,
   i.e., that applying such a function $f$ to a point $x$ cannot increase the density of algorithmic information that it
   contains about a point $y$.  We also survey the conditional  dimensions $\dim(x|y)$ and $\Dim(x|y)$ recently
   developed by the first author and N. Lutz \cite{LutLut17}.  Roughly speaking, these conditional dimensions quantify the
   density of algorithmic information in $x$ beyond what is already present in $y$.

   It is rare for the theory of computing to be used to answer open questions in mathematical analysis whose statements do
   not involve computation or related aspects of logic.  In Section \ref{sec:5}\ we survey exciting new developments that do exactly
   this.  We first describe new characterizations by the first author and N. Lutz \cite{LutLut17}\ of the classical Hausdorff
   and packing dimensions of arbitrary sets in Euclidean spaces in terms of the relativized  dimensions of the
   individual points that belong to them.  These characterizations are called point-to-set principles because they enable
   one to use a bound on the relativized  dimension of a single, judiciously chosen point $x$ in a set $E$ in
   Euclidean space to prove a bound on the classical Hausdorff or packing dimension of the set $E$.  We illustrate the power
   of the point-to-set principle by giving an overview of its use in the new, information-theoretic proof \cite{LutLut17}\ of
   Davies's 1971 theorem stating that the Kakeya conjecture holds in the Euclidean plane \cite{Davi71}.  We then discuss
   two very recent uses of the point-to-set principle to solve open problems in classical fractal geometry.  These are N.
   Lutz and D. Stull's strengthened lower bounds on the Hausdorff dimensions of generalized Furstenberg sets
   \cite{LutStu17}\ and N. Lutz's extension of the fractal intersection formulas for Hausdorff and packing dimensions
   in Euclidean spaces from Borel sets to arbitrary sets.  These are, to the best of our knowledge, the first uses of
   algorithmic information theory to solve open problems in classical mathematical analysis.

   We briefly survey promising directions for future research in Section \ref{sec:6}.  These include extending the algorithmic analysis of
   self-similar fractals \cite{LM08}\ to other classes of fractals,
extending algorithmic dimensions to metric spaces other than
Euclidean spaces, investigating algorithmic fractal dimensions that
are more effective than constructive dimensions (e.g.,
polynomial-time or finite-state fractal dimensions) in fractal
geometry, and extending algorithmic methods to rectifiability and
other aspects of geometric measure theory that do not necessarily
concern fractal geometry.  In each of these we begin by describing
an existing result that sheds light on the promise of further
inquiry.

Overviews of algorithmic dimensions in Cantor space appear in
\cite{DowHir10,EFDAIT}, though these are already out of date. Even
prior to the development of algorithmic fractal dimensions, a rich
network of relationships among gambling strategies, Hausdorff
dimension, and Kolmogorov complexity was uncovered by reserach of
Ryabko \cite{Ryab84,Ryab86,Ryab93,Ryab94}, Staiger
\cite{Stai93,Stai98,Stai00}, and Cai and Hartmanis \cite{CaiHar94}.
A brief account of this ``prehistory'' of algorithmic fractal
dimensions appears in section 6 of \cite{DISS}.

\section{Algorithmic Information in Euclidean Spaces}\label{sec:2}

Algorithmic information theory has most often been used in the set
$\strings$ of all finite binary strings. The {\sl conditional
Kolmogorov complexity\/} (or {\sl conditional algorithmic
information content}) of a string $x\in\strings$ {\sl given\/} a
string $y\in\strings$ is
\begin{equation*}\K(x|y)=\min\myset{|\pi|}{\pi\in\strings\
\mathrm{and}\ U(\pi,y)=x}.\end{equation*} Here $U$ is a fixed
universal Turing machine, and $|\pi|$ is the length of a binary ``
program'' $\pi$. Hence $\K(x|y)$ is the minimum number of bits
required to specify $x$ to $U$, when $y$ is provided as side
information. We refer the reader to any of the standard texts
\cite{LiVit09, DowHir10, Nies12, ShUsVe17}\ for the history and
intuition behind this notion, including its essential invariance
with respect to the choice of the universal Turing machine $U$. The
{\sl Kolmogorov complexity\/} (or {\sl algorithmic information
content}) of a string $x\in \strings$ is then
\begin{equation*}\K(x)=\K(x|\lambda),\end{equation*} where $\lambda$
is the empty string.

Routine binary encoding enables one to extend the definitions of
$\K(x)$ and $\K(x|y)$ to situations where $x$ and $y$ range over
other countable sets such as $\N$, $\Q$, $\N\times \Q$, etc.

The key to ``lifting'' algorithmic information theory notions to
Euclidean spaces is to define the {\sl Kolmogorov complexity\/} of a
set $E\subseteq\Rn$ to be
\begin{equation}\label{equ21}\K(E)=\min\myset{\K(q)}{q\in\Q^n\cap
E}.\end{equation} (Shen and Vereshchagin \cite{SheVer02}\ used a
very similar notion for a very different purpose.) Note that $\K(E)$
is the amount of information required to specify not the set $E$
itself, but rather {\sl some\/} rational point in $E$. In
particular, this implies that \begin{equation*}E\subseteq F\implies
\K(E)\ge \K(F).\end{equation*} Note also that, if $E$ contains no
rational point, then $\K(E)=\infty$.

The {\sl Kolmogorov complexity\/} of a point $x\in\Rn$ {\sl at
precision\/} $r\in\N$ is
\begin{equation}\label{equ22}   \K_r(x)=\K(\B_{2^{-r}}(x)),   \end{equation}
where $\B_{\epsilon}(x)$ is the open ball of radius $\epsilon$ about
$x$, i.e., the number of bits required to specify {\sl some\/}
rational point $q\in\Q^n$ satisfying $|q-x|<2^{-r}$, where $|q-x|$
is the Euclidean distance of $q-x$ from the origin.

\section{Algorithmic Dimensions}\label{sec:3}

   \subsection{Dimensions of Points}\label{subsec:31}
We now define the {\sl (constructive) dimension\/} of a point
$x\in\Rn$ to be \begin{equation}\label{equ31} \dim(x)=\liminf_{r\to
\infty}\frac{\K_r(x)}{r}\end{equation} and the {\sl strong
(constructive) dimension\/} of $x$ to be
\begin{equation}\label{equ32} \Dim(x)=\limsup_{r\to \infty}\frac{\K_r(x)}{r}.\end{equation} We note that $\dim(x)$ and
$\Dim(x)$ were originally defined in terms of algorithmic betting
strategies called gales \cite{DISS,ESDAICC}. The identities
(\ref{equ31}) and (\ref{equ32}) were subsequent {\sl theorems\/}
proven in \cite{LM08}, refining very similar results in
\cite{May02,ESDAICC}. These identities have been so convenient for
work in Euclidean space that it is now natural to regard them as
definitions.

Since $\K_r(x)$ is the amount of information required to specify a
rational point that approximates $x$ to within $2^{-r}$ (i.e., with
$r$ bits of precision), $\dim(x)$ and $\Dim(x)$ are intuitively the
{\sl lower\/} and {\sl upper asymptotic densities of information\/}
in the point $x$. This intuition is a good starting point, but the
fact that $\dim(x)$ and $\Dim(x)$ are {\sl geometrically\/}
meaningful will only become evident in light of the mathematical
consequences of (\ref{equ31}) and (\ref{equ32}) surveyed in this
chapter.

It is an easy exercise to show that, for all $x\in \Rn$,
\begin{equation}\label{equ33} 0\le \dim(x)\le \Dim(x)\le n.
\end{equation}
If $x$ is a computable point in $\Rn$, then $\K_r(x)=o(r)$, so
$\dim(x)=\Dim(x)=0$. On the other hand, if $x$ is a random point in
$\Rn$ (i.e., a point that is algorithmically random in the sense of
Martin-L\"of \cite{MarL66}), then $\K_r(x)=nr-O(1)$, so
$\dim(x)=\Dim(x)=n$. Hence the dimensions of points range between 0
and the dimension of the Euclidean space that they inhabit. In fact,
for {\sl every\/} real number $\alpha \in [0,n]$, the {\sl dimension
level
set}\begin{equation}\label{equ34}\DIM^{\alpha}=\myset{x\in\Rn}{\dim(x)=\alpha}\end{equation}
and the {\sl strong dimension level set}
\begin{equation}\label{equ35}\DIM_{\str}^{\alpha}=\myset{x\in\Rn}{\Dim(x)=\alpha}\end{equation} are uncountable and dense in $\Rn$
\cite{DISS,ESDAICC}. The dimensions $\dim(x)$ and $\Dim(x)$ can
coincide, but they do not generally do so. In fact, the set
$\DIM^0\cap \DIM_{\str}^{n}$ is a comeager (i.e., topologically
large) subset of $\Rn$ \cite{HitPav05}.

Classical fractal geometry has {\sl local}, or {\sl pointwise},
dimensions that are useful, especially in connection with dynamical
systems. Specifically, if $\nu$ is an {\sl outer measure\/} on
$\Rn$, i.e., a function $\nu: \mathcal{P}(\Rn) \to [0, \infty]$
satisfying $\nu(\emptyset)=0$, monotonicity ($E\subseteq F \implies
\nu(E)\le \nu(F)$), and countable subadditivity ($E\subseteq
\cup_{k=0}^{\infty} E_k \implies \nu(E)\le \sum_{k=0}^{\infty}
\nu(E_k)$), and if $\nu$ is {\sl locally finite} (i.e., every $x\in
\Rn$ has a neighborhood $N$ with $\nu(N)< \infty$), then the {\sl
lower\/} and {\sl upper local dimensions\/} of $\nu$ at a point
$x\in\Rn$ are
\begin{equation}\label{equ36}
(\dimloc\nu)(x)=\liminf_{r\to
\infty}\frac{\log(\frac{1}{\nu(\B_{2^{-r}}(x))})}{r}
\end{equation}
and
\begin{equation}\label{equ37}
(\Dimloc\nu)(x)=\limsup_{r\to
\infty}\frac{\log(\frac{1}{\nu(\B_{2^{-r}}(x))})}{r},
\end{equation}
respectively, where $\log=\log_2$ \cite{Falc14}.

Until very recently, no relationship was known between the
 dimensions $\dim(x)$ and $\Dim(x)$ and the local
dimensions (\ref{equ36}) and (\ref{equ37}). However, N. Lutz
recently observed that a very non-classical choice of the outer
measure $\nu$ remedies this. For each $E\subseteq \Rn$, let
\begin{equation}\label{equ38}\kappa(E)=2^{-\K(E)},\end{equation} where $\K(E)$ is defined as in (\ref{equ21}). Then $\kappa$
is easily seen to be an outer measure on $\Rn$ that is {\sl finite}
(i.e., $\kappa(\Rn)<\infty$), hence certainly locally finite, whence
the local dimensions $\dimloc \kappa$ and $\Dimloc \kappa$ are well
defined. In fact we have the following.

\begin{theorem}{(N. Lutz\cite{NLut16})} For all $x\in\Rn$,
\begin{equation*}\dim(x)=(\dimloc \kappa)(x)\end{equation*} and \begin{equation*}\Dim(x)=(\Dimloc \kappa)(x).\end{equation*}\end{theorem}

There is a direct conceptual path from the classical Hausdorff and
packing dimensions to the dimensions of points defined in
(\ref{equ31}) and (\ref{equ32}).

The {\sl Hausdorff dimension\/} $\dimh(E)$ of a set $E\subseteq \Rn$
was introduced by Hausdorff \cite{Haus19}\ before 1920 and is
arguably the most important notion of fractal dimension. Its
classical definition, which may be found in standard texts such as
\cite{SteSha05, Falc14, BisPer17}, involves covering the set $E$ by
families of sets with diameters vanishing in the limit. In all
cases, $0\le \dimh(E)\le n$.

At the beginning of the present century, in order to formulate
versions of Hausdorff dimensions that would work in complexity
classes and other algorithmic settings, the first author \cite{DCC}\
gave a new characterization of Hausdorff dimension in terms of
betting strategies, called gales, on which it is easy to impose
computability and complexity conditions. Of particular interest
here, he then defined the {\sl constructive dimension\/} $\adim(E)$
of a set $E\subseteq \Rn$ exactly like the gale characterization of
$\dimh(E)$, except that the gales were now required to be lower
semicomputable \cite{DISS}. He then defined the dimension $\dim(x)$
of a point $x\in\Rn$ to be the constructive dimension of its
singleton, i.e., $\dim(x)=\adim(\{x\})$. The existence of a
universal Turing machine made it immediately evident that
constructive dimension has the {\sl absolute stability\/} property
that
\begin{equation}\label{equ39} \adim(E)= \sup_{x\in E}
\dim(x)\end{equation} for all $x\in\Rn$. Accordingly, constructive
dimension has since been investigated pointwise. As noted earlier,
the second author \cite{May02}\ then proved the characterization
(\ref{equ31}) as a theorem.

Two things should be noted about the preceding paragraph. First,
these early papers were written entirely in terms of binary
sequences, rather than points in Euclidean space. However,  the most
straightforward binary encoding of points bridges this gap. (In this
survey we freely use those results from Cantor space that do extend
easily to Euclidean space.) Second, although the gale
characterization is essential for polynomial time and many other
stringent levels of effectivization, constructive dimension can be
defined equivalently by effectivizing Hausdorff's original
formulation \cite{Reim04}.


\subsection{The Correspondence Principle}

In 2001, the first author conjectured that there should be a {\sl
correspondence principle\/} (a term that Bohr had used analogously
in quantum mechanics) assuring us that for sufficiently simple sets
$E\subseteq \Rn$, the constructive and classical dimensions agree,
i.e.,
\begin{equation}\label{equ310}\adim(E)=\dimh(E).\end{equation}
Hitchcock \cite{Hitchcock:CPED}\ confirmed this conjecture, proving
that (\ref{equ310}) holds for any set $E\subseteq \Rn$ that is a
union of sets that are computably closed, i.e., that are $\Pi^0_1$
in Kleene's arithmetical hierarchy. (This means that (\ref{equ310})
holds for all $\Sigma^0_2$ sets, and also for sets that are
nonuniform unions of  $\Pi^0_1$  sets.) Hitchcock also noted that
this result is the best possible in the arithmetical hierarchy,
because there are  $\Pi^0_2$ sets $E$ (e.g., $E=\{z\}$, where $z$ is
a Martin-L\"of random point that is $\Delta^0_2$) for which
(\ref{equ310}) fails.

By (\ref{equ39}) and (\ref{equ310}) we have
\begin{equation}\label{equ311}\dimh(E)=\sup_{x\in E}\dim(x),
\end{equation} which is a very nonclassical, pointwise
characterization of the classical Hausdorff dimensions of sets that
are unions of $\Pi^0_1$ sets. Since most textbook examples of
fractal sets are $\Pi^0_1$, (\ref{equ311}) is a strong preliminary
indication that the dimensions of points are geometrically
meaningful.

The {\sl packing dimension \/} $\dimpack(E)$ of a set $E\subseteq
\Rn$ was introduced in the early 1980s by Tricot \cite{Tric82}\ and
Sullivan \cite{Sull84}. Its original definition is a bit more
involved that that of Hausdorff dimension \cite{Falc14,BisPer17}\
and implies that $\dimh(E)\le \dimpack(E)\le n$ for all
$E\subseteq\Rn$.

After the development of constructive versions of Hausdorff
dimension outlined above, Athreya,  Hitchcock, and the authors
\cite{ESDAICC}\ undertook an analogous development for packing
dimension. The gale characterization of $\dimpack(E)$ turns out to
be exactly dual to that of $\dimh(E)$, with just one limit superior
replaced by a limit inferior. The {\sl strong constructive
dimension\/} $\aDim(E)$ of a set $E\subseteq \Rn$ is defined by
requiring the gales to be lower semicomputable, and the {\sl strong
dimension\/} of a point $x\in\Rn$ is $\Dim(x)=\aDim(\{x\})$. The
{\sl absolute stability\/} of strong constructive dimension,
\begin{equation}\label{equ312}\aDim(E)=\sup_{x\in
E}\Dim(x),\end{equation} holds for all $E\subseteq \Rn$, as does the
Kolmogorov complexity characterization (\ref{equ32}). All this was
shown in \cite{ESDAICC}, but a correspondence principle for strong
constructive dimension was left open. In fact, Conidis
\cite{Coni08}\ subsequently used a clever priority argument to
construct a $\Pi^0_1$ set $E\subseteq \Rn$ for which $\aDim(E)\ne
\dimpack(E)$. It is still not known whether some simple, logical
definability criterion for $E$ implies that $\aDim(E)=\dimpack(E)$.
Staiger's proof that regular $\omega$-languages $E$ satisfy this
identity is an encouraging step in this direction \cite{Stai07}.


   \subsection{Self-Similar Fractals}\label{subsec:33}

The first application of algorithmic dimensions to fractal geometry
was the authors' investigation of the dimensions of points in
self-similar fractals \cite{LM08}. We give  a brief exposition of
this work here, referring the reader to \cite{LM08}\ for the many
missing details.

Self-similar fractals are the most widely known and best understood
classes of fractals \cite{Falc14}. Cantor's middle-third set, the
von Koch curve, the Sierpinski triangle, and the Menger sponge are
especially well known examples of self-similar fractals.

Briefly, a self-similar fractal in a Euclidean space $\Rn$ is
generated from an initial nonempty closed set $D\subseteq \Rn$ by an
{\sl iterated function system (IFS)}, which is a finite list
$S=(S_0, S_1, \ldots, S_{k-1})$ of $k\ge 2$ contracting similarities
$S_i: D\to D$. Each of these similarities $S_i$ is {\sl coded\/} by
the symbol $i$ in the alphabet $\Sigma=\{0, \ldots, k-1\}$, and each
$S_i$ has a {\sl contraction ratio\/} $c_i\in(0,1)$. The IFS $S$ is
required to satisfy Moran's {\sl open set condition} \cite{Mora46},
which says that there is a nonempty open set $G\subseteq D$ whose
images $S_i(G)$, for $i\in\Sigma$, are disjoint subsets of $G$.

For example, the Sierpinski triangle is generated from the set
$D\subseteq\R^2$ consisting of the triangle with vertices
$v_0=(0,0)$, $v_1=(1,0)$, and $v_2=(1/2, \sqrt{3}/2)$, together with
this triangle's interior, by the IFS $S=(S_0, S_1, S_2)$, where each
$S_i: D\to D$ is defined by
\begin{equation*}S_i(p)=v_i+\frac{1}{2}(p-v_i)\end{equation*} for $p\in D$. Note
that $\Sigma=\{0,1,2\}$ and $c_0=c_1=c_2=1/2$ in this example. Note
also  that the open set condition is satisfied here by letting $G$
be the topological interior of $D$. Each infinite sequence
$T\in\Sigma^{\infty}$ {\sl codes\/} a point $S(T)\in D$ that is
obtained by applying the similarities coded by the successive
symbols in $T$ in a canonical way. (See Figure \ref{fig1}.) The Sierpinski is
the {\sl attractor\/} (or {\sl invariant set}) of $S$ and $D$, which
consists of all points $S(T)$ for $T\in\Sigma^{\infty}$.

\begin{figure}[htbp]\begin{center}\includegraphics[scale=0.08]{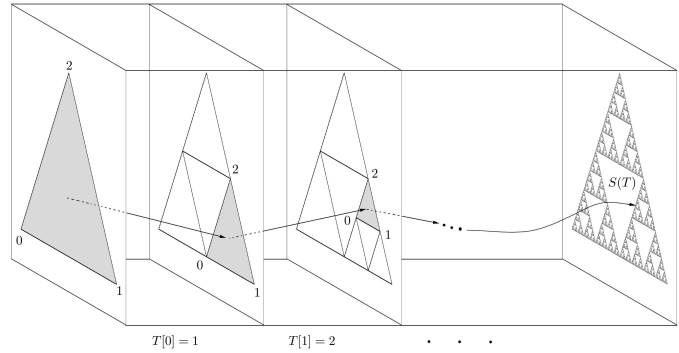}\caption{\label{fig1} A sequence $T\in\{0, 1, 2\}^{\infty}$ codes a point $S(T)$ in the
Sierpinski triangle (from \cite{LM08}).}\end{center}
\end{figure}

The main objective of \cite{LM08}\ was to relate the dimension and
strong dimension of each point $S(T)\in\Rn$ in a self-similar
fractal to the corresponding dimensions of the coding sequence $T$.
As it turned out, the algorithmic dimensions in $\Sigma^{\infty}$
had to be extended in order to achieve this.

The {\sl similarity dimension\/} of an IFS $S=(S_0, \ldots,
S_{k-1})$ with contraction ratios $c_0, \ldots, c_{k-1}\in (0,1)$ is
the unique solution $\sdim(S)=s$ of the equation
\begin{equation}\label{equ313}
\sum_{i=o}^{k-1}c_i^s=1.\end{equation} The {\sl similarity
probability measure\/} of $S$ is the probability measure on $\Sigma$
that is implicit in (\ref{equ313}), i.e., the function $\pi_S:\Sigma
\to [0,1]$ defined by \begin{equation}\label{equ314}\pi_S(i)=
c_i^{\sdim(S)}\end{equation}
 for each $i\in
\Sigma$. If the contraction ratios of $S$ are all the same, then
$\pi_S$ is the uniform probability measure on $\Sigma$, but this is
not generally the case. We extend $\pi_S$ to the domain $\Sigma^*$
by setting
 \begin{equation}\label{equ315}\pi_S(w)=
\prod_{m=0}^{|w|-1}\pi_S(w[m])\end{equation} for each
$w\in\Sigma^*$. We define the Shannon {\sl $S$-self-information\/}
of each string $w\in\Sigma^*$ to be the quantity
 \begin{equation}\label{equ316} l_{S}(w) = \log \frac{1}{ \pi_S(w)}.\end{equation}
Finally, we define the {\sl dimension\/} of a sequence
$T\in\Sigma^{\infty}$ {\sl with respect to\/} the IFS $S$ to be
\begin{equation}\label{equ317}
  \dim^S(T) =    \liminf_{j\to\infty}\frac{\K(T[0..j])}{l_{S}(T[0..j])}.
  \end{equation}
 Similarly,  the {\sl strong  dimension\/} of $T$
{\sl with respect to $S$\/} is
\begin{equation}\label{equ318}
   \Dim^S(T) =    \limsup_{j\to\infty}\frac{\K(T[0..j])}{l_{S}(T[0..j])}.\end{equation}

The dimension (\ref{equ317}) is a special case of an algorithmic
{\sl Billingsley dimension\/} \cite{Bill60, Wegm68, Caja82}. These
are treated more generally in  \cite{LM08}.

A set $F\subseteq\Rn$ is a {\sl computably self-similar fractal\/}
if it is the attractor of some $D$ and $S$ as above such that the
contracting similarities $S_0, \ldots, S_{k-1}$ are all computable
in the sense of computable analysis.

The following theorem gives a complete analysis of the dimensions of
points in computably self-similar fractals.

\begin{theorem}\label{theo32}{(J. Lutz and Mayordomo \cite{LM08})}
If $F\subseteq\Rn$ is a computably self-similar fractal and  $S$ is
an IFS testifying to this fact, then, for all points $x\in F$ and
all coding sequences  $T\in\Sigma^{\infty}$ for $x$,
\begin{equation}\label{equ319}\dim(x)=\sdim(S)
\dim^{S}(T)\end{equation} and
\begin{equation}\label{equ320}\Dim(x)=\sdim(S) \Dim^{S}(T).\end{equation}
\end{theorem}

The proof of Theorem \ref{theo32} is nontrivial. It combines some
very strong coding properties of iterated function systems with some
geometric Kolmogorov complexity arguments.

The following characterization of continuous functions on the reals
is one of the oldest and most beautiful theorems of computable
analysis.

\begin{theorem}\label{the33}{(Lacombe \cite{Lac55a, Lac55b})} A function $f:
\Rn\to \R^m$ is continuous if and only if there is an oracle
$A\subseteq \N$ relative to which $f$ is computable.\end{theorem}

Using Lacombe's theorem it is easy to derive the classical analysis
of self-similar fractals (which need not be computably self-similar)
from Theorem \ref{theo32}.

\begin{corollary}\label{cor34}{(Moran \cite{Mora46}, Falconer \cite{Falc89})}
 For every self-similar
fractal $F\subseteq\Rn$ and every IFS $S$ that generates $F$,
\begin{equation}\label{equ321}\dimh(F)=\dimpack(F)=\sdim(F).
\end{equation}
\end{corollary}

\begin{proof}
Let $F$ and $S$ be as given. By Lacombe's theorem there is an
 oracle
$A\subseteq \N$ relative to which $S$ is computable. It follows by a
theorem by Kamo and Kawamura \cite{KK99}\ that the set $F$ is
$\Pi^0_1$ relative to $A$, whence the relativization of
(\ref{equ311}) tells us that
\begin{equation}\label{equ322}
\dimh^A(F)=\sup_{x\in F}\dim^A(x).\end{equation} We then have

\begin{eqnarray*}
\dimh(F)&\le&\dimpack(F)\\&=&\dimpack^A(F)\\
&\le&\cDim^A(F)\\&=&\sup_{x\in F}\Dim^A(x)\\
&=^{(\ref{equ320})}&\sup_{T\in\Sigma^{\infty}}\sdim(S)\Dim^{S,A}(T)\\
&=&\sdim(S) \\
&=&\sup_{T\in\Sigma^{\infty}}\sdim(S)\dim^{S,
A}(T)\\
&=^{(\ref{equ319})}&\sup_{x\in F}\dim^A(x)\\
&=^{(\ref{equ322})}&\dimh^A(F)\\
&=&\dimh(F),\end{eqnarray*} so (\ref{equ321}) holds.
\end{proof}

Intuitively, Theorem \ref{theo32}\ is stronger than its Corollary
\ref{cor34}, because Theorem \ref{theo32}\ gives a complete account
of ``where the dimension comes from''.

\subsection{Dimension Level Sets}\label{subsec:35}

The dimension level sets $\DIM^{\alpha}$ and  $\DIM_{\str}^{\alpha}$
defined in (\ref{equ34}) and (\ref{equ35}) have been the focus of
several investigations. It was shown in \cite{DISS, ESDAICC}\ that,
for all $0\le \alpha\le n$,
\begin{equation*}\cdim(\DIM^{\alpha})=\dimh(\DIM^{\alpha})=\alpha\end{equation*} and
\begin{equation*}\cDim(\DIM_{\str}^{\alpha})=\dimpack(\DIM_{\str}^{\alpha})=\alpha.\end{equation*}

Hitchcock, Terwijn, and the first author \cite{HiLuTe07}\
investigated the complexities of these dimension level sets from the
viewpoint of descriptive set theory. Following standard usage
\cite{Mosc80}, we write $\bsigk$ and $\bpik$ for the classes at the
$k$th level ($k\in\Z^+$) of the Borel hierarchy of subsets of $\Rn$.
That is, $\bsig{1}$ is the class of all open subsets of $\Rn$, each
$\bpik$ is the class of all complements of sets in $\bsigk$, and
each $\bsig{k+1}$ is the class of all countable unions of sets in
$\bpik$. We also write $\Sigma_k^0$ and $\Pi_k^0$ for the classes of
the $k$th level of Kleene's arithmetical hierarchy of subsets of
$\Rn$. That is, $\Sigma_1^0$ is the class of all computably open
subsets of $\Rn$, each $\Pi_k^0$ is the class of all complements of
sets in $\Sigma_k^0$, and each $\Sigma_{k+1}^0$ is the class of all
effective (computable)
 unions
of sets in $\Pi_k^0$.

Recall that a real number $\alpha$ is {\sl$\Delta_2^0$-computable\/}
if there is a computable function $f: \N\to \Q$ such that
$\lim_{k\to\infty}f(k)=\alpha$.

The following facts were proven in \cite{HiLuTe07}.
\begin{enumerate}
\item $\DIM^0$ is $\Pi^0_2$ but not $\bsig{2}$.
\item For all $\alpha\in(0,n]$, $\DIM^{\alpha}$ is $\bpi{3}$ (and
$\Pi^0_3$ if $\alpha$ is $\Delta^0_2$-computable) but not
$\bsig{3}$.
\item $\DIM^n_{\str}$ is $\bpi{2}$ and $\Pi^0_3$ but not $\bsig{2}$.

\item   For all $\alpha\in[0,n)$, $\DIM^{\alpha}_{\str}$ is $\bpi{3}$  (and
$\Pi^0_4$ if $\alpha$ is $\Delta^0_2$-computable) but not
$\bsig{3}$.

\end{enumerate}

Weihrauch and the first author \cite{LutWei08}\ investigated the
connectivity properties of sets of the form
\begin{equation*}\DIM^I=\bigcup_{\alpha\in I}\DIM^{\alpha}, \end{equation*} where $I\subseteq
[0,n]$ is an interval. After making the easy observation that each
of the sets $\DIM^{[0,1)}$ and $\DIM^{(n-1,n]}$ is totally
disconnected, they proved that each of the sets $\DIM^{[0,1]}$ and
$\DIM^{[n-1,n]}$ is path-connected. These results are especially
intriguing in the Euclidean plane, where they say that extending
either of the sets $\DIM^{[0,1)}$ and $\DIM^{(1,2]}$ to include the
level set $\DIM^1$ transforms it from a totally disconnected set to
a path-connected set. This suggests that $\DIM^1$ is somehow a very
special subset of $\R^2$.

Turetsky \cite{Ture11}\ investigated this matter further and proved
that $\DIM^1$ is a connected set in $\Rn$. He also proved that
$\DIM^{[0,1)}\cup\DIM^{(1,2]}$ is not a path-connected subset of
$\R^2$.

   \subsection{ Dimensions of Points on Lines}




Since effective dimension is a pointwise property, it is natural to
study the dimension spectrum of a set $E\subseteq\Rn$, i.e., the set
$\spec(E)=\myset{\dim(x)}{x\in E}$. This study is far from obvious
even for sets as apparently simple as straight lines. We review in
this section the results obtained so far, mainly for the case of
straight lines in $\R^2$.

As noted in section \ref{subsec:35}, the set of points in $\R^2$ of
dimension exactly one is connected, while the set of points in
$\R^2$ with dimension less than 1 is totally disconnected. Therefore
every line in $\R^2$ contains a point of dimension 1. Despite the
surprising fact that there are lines in every direction that contain
no random points \cite{LL15}, the first author and N. Lutz have
shown that almost every point on any line with random slope has
dimension 2 \cite{LutLut17}. Still all these results leave open
fundamental questions about the structure of the dimension spectra
of lines, since they don't even rule out the possibility of a line
having the singleton set $\{1\}$ as its dimension spectrum.

Very recently this latest open question has been answered in the
negative. N. Lutz and Stull \cite{LutStu17}\ have proven the
following general lower bound on the dimension of points on lines in
$\R^2$.

\begin{theorem}\label{the401}{(N. Lutz and Stull \cite{LutStu17})}
For all $a,b,x\in\R$, \begin{equation*}\dim(x,ax+b)\ge
\dim^{a,b}(x)+\min\{\dim(a,b), \dim^{a,b}(x)\}.\end{equation*} In
particular, for almost every $x\in \R$, $\dim(x,
ax+b)=1+\min\{\dim(a,b),1\}$.\end{theorem}

Taking $x_1=0$ and $x_2$ a Martin-L\"of random real relative to
$(a,b)$, Theorem \ref{the401}\ gives us two points in the line, $(0,
b)$ and $(x_2, ax_2+b)$, whose dimensions differ by at least one, so
the dimension spectrum cannot be a singleton.

We briefly sketch here the main intuitions behind the (deep) proof
of Theorem \ref{the401}, fully based on algorithmic information
theory. Theorem \ref{the401}'s aim is to connect $\dim(x, ax+b)$
with $\dim(a,b,x)$ (i.e., a dimension in $\R^2$ with a dimension in
$\R^3$). Notice that in the case $\dim(a,b)\le \dim^{a,b}(x)$ the
theorem's conclusion is close to saying $\dim(x, ax+b)\ge
\dim(a,b,x)$.

The proof is based on the property that says that under the
following two conditions
\begin{enumerate}
\renewcommand{\labelenumi}{(\roman{enumi}) }
\item\label{c1} $\dim(a,b)$ is small
\item\label{c2} whenever $ux+v=ax+b$, either $\dim(u,v)$ is large or $(u,v)$ is close to $(a,b)$
\end{enumerate}
it holds that $\dim(x, ax+b)$ is close to $\dim(a,b,x)$.

There is an extra ingredient to finish this intuition.
While condition (ii) can be shown to hold in general, condition (i)
can only be proven in a particular relativized world whereas the
conclusion of the theorem still holds for every oracle.

 N. Lutz and Stull \cite{LS17}\ have also shown that the dimension
spectrum of a line is always infinite, proving the following two
results. The first theorem proves that if $\dim(a,b)=\Dim(a,b)$ then
the corresponding line contains a length one interval.

\begin{theorem}\label{the402}{(N. Lutz and Stull \cite{LS17})}
Let $a,b\in\R$ satisfy that $\dim(a,b)=\Dim(a,b)$. Then for every
$s\in [0,1]$ there is a point $x\in\R$ such that
$\dim(x,ax+b)=s+\min\{\dim(a,b), 1\}$. \end{theorem}

The second result proves that all spectra of lines are infinite.

\begin{theorem}\label{the43}{(N. Lutz and Stull \cite{LS17})}
Let $L_{a,b}$ be any line in $\R^2$. Then the dimension spectrum
$\spec(L_{a,b})$ is infinite.
\end{theorem}


\section{Mutual and Conditional Dimensions}\label{sec:4}


Just as the dimension of a point $x$ in Euclidean space is the
asymptotic density of the algorithmic information in $x$, the mutual
dimension between two points $x$ and $y$ in Euclidean spaces is the
asymptotic density of the algorithmic information shared by $x$ and
$y$. In this section, we survey this notion and the data processing
inequalities, which estimate the effect of computable functions on
mutual dimension. We also survey the related notion of conditional
dimension.

\subsection{Mutual Dimensions}\label{subsec:41} The {\sl mutual (algorithmic) information\/} between two rational points $p\in\Q^m$ and
$q\in\Q^n$ is
\begin{equation*}\miI(p:q)=\K(p)-\K(p|q).\end{equation*} This notion, essentially due to Kolmogorov \cite{Kolm65}, is an analog of mutual entropy in
Shannon information theory \cite{Shan48,CovTho06, LiVit09}.
Intuitively, $\K(p|q)$ is the amount of information in $p$ {\sl
not\/} contained in $q$, so $\miI(p:q)$ is the amount of information
in $p$ that {\sl is\/} contained in $q$. It is well known
\cite{LiVit09}\ that, for all $p\in\Q^m$ and $q\in\Q^n$,
\begin{equation}\label{equ51n}\miI(p:q)\approx
\K(p)+\K(q)-\K(p,q)\end{equation} in the sense that the magnitude of
the difference between the two sides of (\ref{equ41})\ is
$o(\min\{\K(p), \K(q)\})$. This fact is called {\sl symmetry of
information}, because it immediately implies that $\miI(p:q)\approx
\miI(q:p)$.

The ideas in the rest of this section were introduced by Case and
the first author \cite{CasLut15}. In the spirit of (\ref{equ21})
they defined the {\sl mutual information\/} between sets
$E\subseteq\R^m$ and $F\subseteq\Rn$ to be
\begin{equation*}\miI(E:F)=\min\myset{\miI(p:q)}{p\in\Q^m\cap E\ \mathrm{and}\
q\in\Q^n\cap F}.\end{equation*} This is the amount of information
that rational points $p$ and $q$ {\sl must\/} share in order to be
in $E$ and $F$, respectively. Note that, for all $E_1,
E_2\subseteq\R^m$ and $F_1, F_2\subseteq\Rn$,
\begin{equation*}\bigl[(E_1\subseteq E_2)\ \mathrm{and}\
(F_1\subseteq F_2)\bigr]\implies \miI(E_1:F_1)\ge
\miI(E_2:F_2).\end{equation*}

The {\sl mutual information\/} between two points $x\in\R^m$ and
$y\in\Rn$ at {\sl precision\/} $r\in\N$ is
\begin{equation*}\miI_r(x:y)=\miI(\B_{2^{-r}}(x): \B_{2^{-r}}(y)).\end{equation*} This is the amount
of information that rational approximations of $x$ and $y$ {\sl
must\/} share, merely due to their proximities (distance less than
$2^{-r}$) to $x$ and $y$.

In analogy with (\ref{equ31}) and (\ref{equ32}), the {\sl lower\/}
and {\sl upper  mutual dimensions\/} between points $x\in\R^m$ and
$y\in\R^m$ are
\begin{equation}\label{equ41}\mdim(x:y)=\liminf_{r\to\infty}\frac{\miI_r(x:y)}{r}\end{equation}
and
\begin{equation}\label{equ42}\Mdim(x:y)=\limsup_{r\to\infty}\frac{\miI_r(x:y)}{r},\end{equation}
respectively.

The following theorem shows that the  mutual dimensions mdim and
Mdim have many of the properties that one should expect them to
have. The proof is involved and includes a modest generalization of
Levin's coding theorem \cite{Levi73a, Levi74}.

\begin{theorem}\label{the41}{(Case and J. Lutz \cite{CasLut15})} For
all $x\in\R^m$ and $y\in\Rn$, the following hold.
\begin{enumerate}
\item $\mdim(x:y)\le \min\{\dim(x), \dim(y)\}$.
\item $\Mdim(x:y)\le \min\{\Dim(x), \Dim(y)\}$.
\item $\mdim(x:x)=\dim(x)$.
\item $\Mdim(x:x)=\Dim(x)$.
\item $\mdim(x:y)=\mdim(y:x)$.
\item $\Mdim(x:y)=\Mdim(y:x)$.
\item\label{the417} $\dim(x)+\dim(y)-\Dim(x,y)\le \mdim(x:y)\le
\Dim(x)+\Dim(y)-\Dim(x:y)$.
\item\label{the418} $\dim(x)+\dim(y)-\dim(x,y)\le \Mdim(x:y)\le
\Dim(x)+\Dim(y)-\dim(x:y)$.
\item\label{the419} If $x$ and $y$ are independently random, then $\Mdim(x:y)=0$.

\end{enumerate}
\end{theorem}

The expressions $\dim(x,y)$ and $\Dim(x,y)$ in \ref{the417}\ and
\ref{the418}\ above refer to the  dimensions of the point
$(x,y)\in\R^{m+n}$. In \ref{the419}\ above, $x$ and $y$ are {\sl
independently random\/} if $(x,y)$ is a Martin-L\"of random point in
$\R^{m+n}$.

More properties of mutual dimensions may be found in
\cite{CasLut15,CasLut15b}.

   \subsection{Data Processing Inequalities} The data processing
   inequality of Shannon information theory \cite{CovTho06}\ says
   that, for any two probability spaces $X$ and $Y$, any set
   $Z$, and any function $f: X\to Z$, \begin{equation}\label{equ43}
   \miI(f(X);Y)\le \miI(X;Y),\end{equation} i.e., the induced
   probability space $f(X)$ obtained by ``processing the information
   in $X$ through $f$'' has no greater mutual entropy with $Y$ than  $X$
   has with $Y$. More succintly, $f(X)$ tells us no more about $Y$
   than $X$ tells us about $Y$. The data processing inequality of
   algorithmic information theory \cite{LiVit09}\ says that, for any
   computable partial function $f: \strings\to \strings$, there is a
   constant $c_f\in\N$ (essentially the number of bits in a program
   that computes $f$) such that, for all strings $x\in\dom f$ and
   $y\in\strings$,
   \begin{equation}\label{equ44}\miI(f(x):y)\le\miI(x:y)+c_f.\end{equation}
   That is, modulo the constant $c_f$, $f(x)$ contains no more
   information about $y$ than $x$ contains about $y$.

   The data processing inequality for the  mutual
   dimension $\mdim$ should say that every nice function $f:\R^m\to
   \Rn$ has the property that, for all $x\in\R^m$ and
   $y\in\R^k$,\begin{equation}\label{equ45}\mdim(f(x):y)\le
   \mdim(x:y).\end{equation}
But what should ``nice'' mean? A nice function certainly should be
computable in the sense of computable analysis \cite{BC06, Ko91,
Wei00}. But this is not enough. For example, there is a function
$f:\R\to\R^2$ that is computable and {\sl space-filling\/} in the
sense that $[0,1]^2\subseteq \range f$ \cite{Saga94,CDM12}. For such
a function, choose $x\in\R$ such that $\dim(f(x))=2$, and let
$y=f(x)$. Then
\begin{eqnarray*}\mdim(f(x):y)&=&\mdim(y:y)\\&=&\dim(y)\\&=&2\\&>&1\\&\ge&\dim(x)\\&\ge&\mdim(x:y),\end{eqnarray*}
so (\ref{equ45}) fails.

Intuitively, the above failure of (\ref{equ45}) occurs because the
function $f$ is extremely sensitive to its input, a property that ``
nice'' functions do not have. A function $f:\R^m\to \Rn$ is {\sl
Lipschitz\/} if there is a real number $c>0$ such that, for all
$x_1, x_2\in\R^m$, \begin{equation*}|f(x_1)-f(x_2)|\le c
|x_1-x_2|.\end{equation*} The following data processing inequalities
show that computable Lipschitz functions are ``nice''.

\begin{theorem}\label{the42}{(Case and J. Lutz \cite{CasLut15})} If
$f:\R^m\to \Rn$ is computable and Lipschitz, then, for all
$x\in\R^m$ and $y\in\R^k$,\begin{equation*}\mdim(f(x):y)\le
\mdim(x:y)\end{equation*} and
\begin{equation*}\Mdim(f(x):y)\le \Mdim(x:y)\end{equation*}\end{theorem}

Several more theorems of this type and applications of these appear
in \cite{CasLut15}.

   \subsection{Conditional Dimensions}
A comprehensive theory of the  fractal dimensions of points in
Euclidean spaces requires not only the  dimensions $\dim(x)$ and
$\Dim(x)$ and the  mutual dimensions $\mdim(x:y)$ and $\Mdim(x:y)$,
but also the  conditional dimensions $\dim(x|y)$ and $\Dim(x|y)$
formulated by the first author and N. Lutz \cite{LutLut17}. We
briefly describe these formulations here.

The conditional Kolmogorov complexity $\K(p|q)$, defined for
rational points $p\in\Q^m$ and $q\in\Q^n$, is lifted to the
conditional  dimensions in the following four steps.

\begin{enumerate}
\item For $x\in\R^m$, $q\in\Q^n$, and $r\in\N$, the {\sl conditional
Kolmogorov complexity of $x$ at precision $r$ given $q$\/} is
\begin{equation*}\hat{\K}_r(x|q)=\min\myset{\K(p|q)}{p\in\Q^m\cap \B_{2^{-r}}(x)}.\end{equation*}
\item For $x\in\R^m$, $y\in\R^n$, and $r,s\in\N$, the {\sl
conditional Kolmogorov complexity\/} of $x$ {\sl at precision $r$
given $y$ at precision $s$ \/} is
\begin{equation*}\K_{r,s}(x|y)=\max\myset{\hat{\K}_r(x|q)}{q\in\Q^n\cap
\B_{2^{-s}}(y)}.\end{equation*}
\item For $x\in\R^m$, $y\in\R^n$, and $r\in\N$, the {\sl
conditional Kolmogorov complexity\/} of $x$ {\sl given $y$\/} at
precision $r$  is
\begin{equation*}\K_{r}(x|y)=\K_{r,r}(x|y).\end{equation*}
\item For $x\in\R^m$ and $y\in\R^n$, the {\sl lower\/} and
{\sl upper  conditional dimensions\/} of {\sl $x$ given $y$ \/} are
\begin{equation*}\dim(x|y)=\liminf_{r\to\infty}\frac{\K_r(x|y)}{r}\end{equation*}
and
\begin{equation*}\Dim(x|y)=\limsup_{r\to\infty}\frac{\K_r(x|y)}{r},\end{equation*}
respectively.

\end{enumerate}

Steps 1, 2, and 4 of the above lifting are very much in the spirit
of what has been done in section \ref{sec:2}, \ref{subsec:31}, and
\ref{subsec:41} above. Step 3 appears to be problematic, because
using the same precision bound $r$ for both $x$ and $y$ makes the
definition seem arbitrary and ``brittle''. However, the following
result shows that this is not the case.

\begin{theorem}{(\cite{LutLut17})} Let $s:\N\to\N$. If
$|s(r)-r|=o(r)$, then, for all $x\in\R^m$ and $y\in\R^n$,
\begin{equation*}\dim(x|y)=\liminf_{r\to \infty}\frac{\K_{r,s(r)}(x|y)}{r}\end{equation*}
and
\begin{equation*}\Dim(x|y)=\limsup_{r\to \infty}\frac{\K_{r,s(r)}(x|y)}{r}.\end{equation*}
\end{theorem}

The following result is useful for many purposes.

\begin{theorem}{(chain rule for $\K_r$)} For all $x\in \R^m$ and
$y\in\Rn$,
\begin{equation}\label{equ3square}
\K_r(x,y) = \K_r(x|y)+\K_r(y)+o(r).
\end{equation}
 \end{theorem}

An {\sl oracle\/} for a point $y\in\R^n$ is a function $g: \N\to
\Q^n$ such that, for all $s\in\N$, $|g(s)-y|\le 2^{-s}$. The {\sl
Kolmogorov complexity\/} of a rational point $p\in\Q^m$ {\sl
relative to \/} a point $y\in\R^n$ is
\begin{equation*}\K^y(p)=\max\myset{\K^g(p)}{g \mbox{ is an oracle for } y},\end{equation*}
where $\K^g(p)$ is the Kolmogorov complexity of $p$ when the
universal machine has access to the oracle $g$. The purpose of the
maximum here is to prevent $\K^y(p)$ from using oracles $g$ that
code more than $y$ into their behaviors. For $x\in\R^m$ and
$y\in\R^n$, the {\sl  dimension $\dim^y(x)$ relative to $y$\/} is
defined from $\K^y(p)$ exactly as $\dim(x)$ was defined from $\K(p)$
in Sections \ref{sec:2}\ and \ref{subsec:31}\ above. The {\sl
relativized strong
 dimension $\Dim^y(x)$\/} is defined analogously.

The following result captures the intuition that conditioning on a
point $y$ is a restricted form of oracle access to $y$.

\begin{lemma}{(\cite{LutLut17})}\label{lem44} For all $x\in\R^m$ and $y\in\R^n$,
$\dim^y(x)\le\dim(x|y)$ and $\Dim^y(x)\le\Dim(x|y)$.
\end{lemma}

The remaining results in this section confirm that conditional
 dimensions have the correct information-theoretic
relationships to  dimensions and  mutual dimensions.

\begin{theorem}{(\cite{LutLut17})}\label{theo45} For all $x\in\R^m$ and $y\in\R^n$,
\begin{equation*}\mdim(x:y)\ge\dim(x)-\Dim(x|y)\end{equation*} and \begin{equation*}\Mdim(x:y)\le\Dim(x)-\dim(x|y).\end{equation*}
\end{theorem}

\begin{theorem}{(chain rule for dimension \cite{LutLut17})}\label{theo46}
For all $x\in\R^m$ and $y\in\R^n$, \begin{eqnarray*} \dim(x)+\dim(y|x)&\le&\dim(x,y)\\
&\le&\dim(x)+\Dim(y|x)\\
&\le&\Dim(x,y)\\
&\le&\Dim(x)+\Dim(y|x).\end{eqnarray*}

\end{theorem}

\section{Algorithmic Discovery of New Classical Theorems}\label{sec:5}


   \subsection{The Point-to-Set Principle}\label{subsec:51} Many of the most challenging problems in geometric measure theory
   are problems of establishing lower bounds on the  classical fractal dimensions $\dimh(E)$ and $\dimpack(E)$ for sets $E\subseteq
   \Rn$. Although such problems seem to involve global properties of the sets $E$ and make no mention of algorithms, the
    dimensions of points have recently been used to prove new lower bound results for classical fractal
   dimensions. The key to these developments is the following pair of  theorems of the first author and N. Lutz.
   \begin{theorem}\label{the51}{(point-to-set-principle for Hausdorff dimension \cite{LutLut17})} For every $E\subseteq \Rn$,
   \begin{equation}\label{equ51}\dimh(E)=\min_{A\subseteq\N}\sup_{x\in E}\dim^A(x).\end{equation}\end{theorem}
    \begin{theorem}\label{the52}{(point-to-set-principle for packing dimension \cite{LutLut17})} For every $E\subseteq \Rn$,
   \begin{equation}\label{equ52}\dimpack(E)=\min_{A\subseteq\N}\sup_{x\in E}\Dim^A(x).\end{equation}\end{theorem}

   The relativized dimensions $\dim^A(x)$ and $\Dim^A(x)$ here are defined by substituting $\K_r^A(x)$ for $\K_r(x)$ in
   (\ref{equ31}) and (\ref{equ32}).

   It is to be emphasized that these two theorems completely characterize $\dimh(E)$ and $\dimpack(E)$ for {\sl all\/} sets $E\subseteq
   \Rn$. These characterizations are called {\sl point-to-set principles\/} because they enable one to use a lower bound on
   the relativized dimension of a single, judiciously chosen point $x\in E$ to establish a lower bound on the classical
   dimension of the set $E$ itself. More precisely, for example, Theorem \ref{the51}\ says that, in order to prove a lower
   bound $\dimh(E)\ge \alpha$, it suffices to show that, for every oracle $A\subseteq \N$ and every $\epsilon>0$, there is
   a point $x\in E$ such that $\dim^A(x)\ge \alpha-\epsilon$. In some cases, it can in fact be shown that, for every oracle
   $A\subseteq \N$, there is a point $x\in E$ such that $\dim^A(x)\ge \alpha$. While the arbitrary oracle $A$ is essential
   for the correctness of such proofs, the discussion below shows that its presence has not been burdensome in applications
   to date.

   \subsection{Plane Kakeya Sets} The first application of the point-to-set principle was not a new theorem, but rather a
   new, information-theoretic proof of an old theorem. We describe this proof here because it illustrates the intuitive
   power of the point-to-set principle.

   A {\sl Kakeya set\/} in $\Rn$ is a set $K\subseteq \Rn$ that contains a unit segment in every direction. Sometime before
   1920, Besicovitch \cite{Besi19, Besi28}\ proved the then-surprising existence of Kakeya sets of Lebesgue measure 0 in
   $\Rn$ for all $n\ge 2$ and asked whether Kakeya sets in $\R^2$ can have dimension less than 2 \cite{Davi71}. The famous
   {\sl Kakeya conjecture\/} (in its most commonly stated form) asserts a negative answer to this and the analogous questions
   in higher dimensions. That is, the Kakeya conjecture says that every Kakeya set in a Euclidean space $\Rn$ has Hausdorff
   dimension $n$. This conjecture holds trivially for $n=1$ and Davies \cite{Davi71}\ proved that it holds for $n=2$. The
   Kakeya conjecture remains an important open problem for $n\ge 3$ \cite{Wolf99, Tao00}.

   Our objective here is to sketch the new proof by the first author and N. Lutz \cite{LutLut17}\ of Davies's theorem, that the
   Kakeya conjecture holds in the Euclidean plane $\R^2$. This proof uses the following lower bound on the dimensions of
   points in a line $y=mx+b$.
\begin{lemma}\label{lem53}{(J. Lutz and N. Lutz \cite{LutLut17})} Let $m\in [0,1]$ and $b\in\R$. For almost every $x\in
[0,1]$,
\begin{equation}\label{equ53}\dim(x, mx+b)\ge\liminf_{r\to\infty}\frac{\K_r(m,b,x)-\K_r(b|m)}{r}.\end{equation}\end{lemma}

We do not prove this lemma here, but note that the proof
relativizes, so the lemma holds relative to every oracle $A\subseteq
\N$.

To prove Davies's theorem, let $\K\subseteq\R^2$ be a Kakeya set. By
the point-to-set principle, fix $A\subseteq \N$ such that
\begin{equation}\label{equ54}\dimh(K)=\sup_{(x,y)\in K}\dim^A(x,y).\end{equation} Fix $m\in[0,1]$ such that
\begin{equation}\label{equ55}\dim^A(m)=1.\end{equation} (This holds for any $m$ that is random relative to $A$.) Since $K$
is Kakeya, there is a unit segment $L\subseteq K$ of slope $m$. Let
$(x_0, y_0)$ be the left endpoint of $L$, let $q\in \Q\cap[x_0,
x_0+1/2]$, and let $L'$ be the unit segment of slope $m$ whose
endpoint is $(x_0-q, y_0)$. Then $L'$ crosses the $y$-axis at the
point $b=mq+y_0$. By Lemma \ref{lem53}\ (relativized to $A$), fix
$x\in [0, 1/2]$ such that
\begin{equation}\label{equ56}\dim^{A,m,b}(x)=1\end{equation} and
\begin{equation}\label{equ57}\dim^A(x, mx+b)\ge\liminf_{r\to\infty}\frac{\K_r^A(m,b,x)-\K_r^A(b|m)}{r}.\end{equation}
(Such $x$ exists, because almost every $x\in[0,1/2]$ satisfies
(\ref{equ56}) and (\ref{equ57}).)

In the language of section \ref{subsec:51}, our ``judiciously chosen
point'' is $(x+q,mx+b)\in L \subseteq K$, and the point-to-set
principle tells us that it suffices to prove that
\begin{equation}\label{equ58}\dim^A(x+q,mx+b)=2.\end{equation} But this is now easy. Since $q$ is rational, (\ref{equ57})
and two applications of the chain rule (\ref{equ3square}) tell us
that
\begin{eqnarray*}
\dim^A(x+q,mx+b)&=&\dim^A(x,mx+b)\\
&\ge&\liminf_{r\to \infty} \frac{\K_r^A(m,b,x)-\K_r^A(b,m)+\K_r^A(m)}{r}\\
&=&\liminf_{r\to \infty} \frac{\K_r^A(x|b,m)+\K_r^A(m)}{r}\\
&\ge&\liminf_{r\to \infty} \frac{\K_r^{A,m,b}(x)}{r}+\liminf_{r\to \infty} \frac{\K_r^A(m)}{r}\\
&=&\dim^{A,m,b}(x)+\dim^A(m),\end{eqnarray*} whence (\ref{equ55})
and (\ref{equ56}) tell us that (\ref{equ58}) holds.

This information-theoretic proof of Davies can be summarized in very
intuitive terms: Because $K$ is Kakeya, it contains a unit segment
$L$ whose slope $m$ has dimension 1 relative to $A$. A rational
shift of $L$ to a unit segment $L'$ crosses the $y$-axis at some
point $b$. Lemma \ref{lem53}\ then gives us a point $(x, mx+b)$ on
$L'$ that has dimension 2 relative to $A$. The point on $L$ from
which $(x, mx+b)$ was shifted lies in $K$ and also has dimension 2
relative to $A$, so $K$ has Hausdorff dimension 2.

The following two sections discuss recent uses of this method to
prove {\sl new\/} theorems in classical fractal geometry.


\subsection{Intersections and Products of Fractals}
We now consider two fundamental, nontrivial, textbook theorems of
fractal geometry. The first, over thirty years old and called the
{\sl intersection formula}, concerns the intersection of one fractal
with a random translation of another fractal.
\begin{theorem}\label{theo54}{(Kahane \cite{Kaha86}, Mattila \cite{Matt84, Matt85})} For all Borel sets $E,F\subseteq \Rn$ and
almost every $z\in\Rn$, \begin{equation*}\dimh(E\cap(F+z))\le
\max\{0, \dimh(E\times F)-n\}.\end{equation*}\end{theorem}

The second theorem, over sixty years old and called the {\sl product
formula}, concerns the product of two fractals.
\begin{theorem}\label{theo55}{(Marstrand \cite{Mars54})} For all $E\subseteq\Rn$ and $F\subseteq\Rn$,\begin{equation*}\dimh(E\times F)\ge
\dimh(E)+\dimh(F).\end{equation*}\end{theorem}

In a recent breakthrough, algorithmic dimension was used to prove
the following extension of the intersection formula from Borel sets
to {\sl all\/} sets. We include the simple (given the machinery that
we have developed) and instructive proof here.
\begin{theorem}\label{the56}{(N. Lutz \cite{NLut17})} For {\sl
all\/} sets $E, F\subseteq\Rn$ and almost every $z\in\Rn$,
\begin{equation}\label{equ59}\dimh(E\cap(F+z))\le\max\{0,
\dimh(E\times F)-n\}.\end{equation}\end{theorem}

\begin{proof}
Let $E, F\subseteq\Rn$ and $z\in\Rn$. The theorem is trivially
affirmed if $F+z$ is disjoint from $E$, so assume not. By the
point-to-set principle, fix an oracle $A\subseteq \N$ such that
\begin{equation}\label{equ510}\dimh(E\times F)=\sup_{(x,y)\in E\times F}\dim^A(x,y).\end{equation}

Let $\epsilon>0$. Since $E\cap (F+z)\ne \emptyset$, the point-to-set
principle tells us that there is a point $x\in E\cap(F+z)$
satisfying
\begin{equation}\label{equ511}\dim^{A,z}(x)>\dimh(E\cap(F+z))-\epsilon.\end{equation}
Now $(x, x-z)\in E\times F$, so (\ref{equ510}), Theorem
\ref{the401}, Lemma \ref{lem44}, and (\ref{equ511}) tell us that
\begin{eqnarray*}\dimh(E\times F)&\ge& \dim^A(x,
x-z)\\&=&\dim^A(x,z)\\&\ge&\dim^A(z)+\dim^A(x|z)\\&\ge&\dim^A(z)+\dim^{A,z}(x)\\&>&\dim^A(z)+\dimh(E\cap(F+z))-\epsilon.
\end{eqnarray*}
Since $\epsilon$ is arbitrary here, it follows that
\begin{equation*}\dimh(E\cap(F+z))\le \dimh(E\times F)-\dim^A(z).\end{equation*} Since almost
every $z\in\Rn$ is Martin-L\"of random relative to $A$ and hence
satisfies $\dim^A(z)=n$, this affirms the theorem.
\end{proof}

The paper \cite{NLut17}\ shows that the same method gives a new
proof of the analog of Theorem \ref{the56}\ for packing dimension.
This result was already known to hold for all sets $E$ and $F$
\cite{Falc94}, but the new proof makes clear what a strong duality
between Hausdorff and packing dimensions is at play in the
intersection formulas.

The paper \cite{NLut17}\ also gives a new, algorithmic proof of the
following known extension of Theorem \ref{theo55}.

\begin{theorem}\label{the57}{(Marstrand \cite{Mars54}, Tricot
\cite{Tric82})} For all $E\subseteq\R^m$ and $F\subseteq\Rn$,
\begin{eqnarray*}\dimh(E)+\dimh(F)&\le&\dimh(E\times
F)\\&\le&\dimh(E)+\dimpack(F)\\&\le&\dimpack(E\times
F)\\&\le&\dimpack(E)+\dimpack(F).\end{eqnarray*}\end{theorem}

This new proof is much simpler than previously known proofs of
Theorem \ref{the57}, roughly as simple as previously known proofs of
the restriction of Theorem \ref{the57}\ to Borel sets. The new proof
is also quite natural, using the point-to-set principle to derive
Theorem \ref{the57}\ from the formally similar Theorem \ref{theo46}.

   \subsection{Generalized Furstenberg Sets} For $\alpha\in(0,1]$, a
   plane set $E\subseteq\R^2$ is said to be {\sl of Furstenberg type
   with parameter\/} $\alpha$ or, more simply, {\sl
   $\alpha$-Furstenberg}, if, for every direction $e\in S^1$ (where
   $S^1$ is the unit circle in $\R^2$), there is a line
   $\mathcal{L}_e$ in direction $e$ such that
   $\dimh(\mathcal{L}_e\cap E)\ge \alpha$.

According to Wolff \cite{Wolf99}, the following well-known bound is
probably due to Furstenberg and Katznelson.

\begin{theorem}\label{the58} For every $\alpha\in(0,1]$, every
$\alpha$-Furstenberg set $E\subseteq\R^2$
satisfies\begin{equation*}\dimh(E)\ge \alpha+\max\{1/2,
\alpha\}.\end{equation*}\end{theorem}

Note that every Kakeya set in the plane is 1-Furstenberg (since it
contains a line segment, which has Hausdorff dimension 1, in every
direction $e\in S^1$), so Davies's theorem follows from the case
$\alpha=1$ of Theorem \ref{the58}. It is an open question -- one
with connections to Falconer's distance conjecture \cite{KatTao01}\
and Kakeya sets \cite{Wolf99} -- whether Theorem \ref{the58}\ can be
improved.

In 2012, Molter and Rela generalized $\alpha$-Furstenberg sets in a
natural way. For $\alpha, \beta\in(0,1]$, a
   set $E\subseteq\R^2$ is  {\sl
   $(\alpha, \beta)$-generalized Furstenberg\/} if there is a set
   $J\subseteq S^1$ such that $\dimh(J)\ge \beta$ and, for every
   $e\in J$, there is a line
   $\mathcal{L}_e$ in direction $e$ such that
   $\dimh(\mathcal{L}_e\cap E)\ge \alpha$. They then proved the
   following lower bound.
   \begin{theorem}\label{the59}{(Molter and Rela \cite{MolRel12})}
   For $\alpha, \beta\in(0,1]$, every
   $(\alpha, \beta)$-generalized Furstenberg set $E\subseteq\R^2$
   satisfies \begin{equation*}\dimh(E)\ge \max\{\beta/2,
   \alpha+\beta-1\}.\end{equation*}\end{theorem}

   Note that every $\alpha$-Furstenberg set is
   $(\alpha,1)$-generalized Furstenberg, so Theorem \ref{the58}\
   follows from the case $\beta=1$ of Theorem \ref{the59}.

   Algorithmic dimensions were recently used to prove the following
   result, which improves Theorem \ref{the59}\ when $\alpha, \beta
   \in (0,1)$ and $\beta<2\alpha$.

   \begin{theorem}\label{the510}{(N. Lutz and Stull
   \cite{LutStu17})} For all $\alpha, \beta\in(0,1]$, every
   $(\alpha, \beta)$-generalized Furstenberg set $E\subseteq\R^2$
   satisfies\begin{equation*}\dimh(E)\ge\alpha+\min\{\beta,\alpha\}.\end{equation*}\end{theorem}

   The proof of Theorem \ref{the510} uses the point-to-set principle
   and Theorem \ref{the401}.

\section{Research Directions}\label{sec:6}

   \subsection{Beyond Self-Similarity  }
In previous sections we have analyzed the dimension of points in
self-similar fractals, but interesting natural examples need more
elaborated concepts that combine self-similarity with random
selection. In \cite{GLMM14}\ Gu, Moser, and the authors started the
more challenging task of analyzing the dimensions of points in
random fractals. They focused on fractals that are randomly selected
subfractals of a given self-similar fractal.

Let $F\subseteq\Rn$ be a computably self-similar fractal as defined
in section \ref{subsec:33}, with $S=(S_0,\dots, S_{k-1})$ the
corresponding IFS, and $\Sigma=\{0, \ldots, k-1\}$. Recall that each
point $x\in F$ has a coding sequence $T\in\Sigma^{\infty}$ meaning
that the point $x$ is obtained by applying the similarities coded by
the successive symbols in $T$.
 We are interested in
certain randomly selected subfractals of the fractal $F$.

 The
specification of a point in such a subfractal can be formulated as
the outcome of an infinite two-player game between a {\sl
selector\/} that selects the subfractal and a {\sl coder\/} that
selects a point within the subfractal. Specifically, the selector
selects $r$ out of the $k$ similarities and this choice depends on
the coder's earlier choices, that is, a {\sl selector\/} is a
function $\sigma:\Gamma^*\rightarrow[\Sigma]^r$ where $[\Sigma]^r$
is the set of all $r$-element subsets of $\Sigma$, alphabet
$\Gamma=\{0, \ldots, r-1\}$ and each element in $\Gamma^*$
represents a coder's earlier history. A {\sl coder\/} is a sequence
$U\in\Gamma^\infty$, that is, the coder is selecting a point in the
subfractal by repeatedly choosing a similarity out of the $r$
previously picked by the selector. Once a selector $\sigma$ and a
coder $U$ have been chosen, the outcome of the selector-coder game
is a point determined by the sequence $\sigma *U\in\Sigma^\infty$,
that can be precisely defined as
\begin{equation*}(\sigma
*U)[t]= \mbox{``the $U[t]$th element of $\sigma(U[0 .. t-1])$''}\end{equation*}
for all $t\in\N$.

Each selector $\sigma$  specifies (selects) the subfractal
$F_\sigma$ of $F$ consisting of all points with coding sequence $T$
for which $T$ is an outcome of playing $\sigma$ against some coder,
$F_\sigma =\myset{S(\sigma*U)}{U\in\Gamma^\infty}$.

The focus of \cite{GLMM14}\ is in randomly selected subfractals of
$F$, by which we mean subfractals $F_\sigma$ of $F$ for which the
selector $\sigma$ is random with respect to some probability
measure. That is, we are interested in the case where the coder is
playing a ``game against nature'' (in order to make precise the idea
of algorithmically random selector
 each selector $\sigma:\Gamma^*\rightarrow [\Sigma]^r$ is identified with its {\sl characteristic\/} sequence $\chi_\sigma\in
([\Sigma]^r)^\infty$).

Gu et al. determine the dimension spectra of a wide class of such
randomly selected subfractals, showing that each such fractal has a
dimension spectrum that is a closed interval whose endpoints can be
computed or approximated from the parameters of the fractal. In
general, the maximum of the spectrum is determined by the degree to
which the coder can {\sl reinforce} the randomness in the selector,
while the minimum is determined by the degree to which the coder can
{\sl cancel } randomness in the selector. This randomness
cancellation phenomena has also arisen in other contexts, notably
dimension spectra of random closed sets \cite{BBCDW07,DiaKjo09} and
of random translations of the Cantor set \cite{DLMT12}. The main
result in \cite{GLMM14}\ concerns subfractals that are similarity
random,  that is, $F_\sigma$ defined by a selector $\sigma$ that is
$\hat{\pi_S}$-random. Here $\hat{\pi_S}$ is the natural extension of
$\pi_s$, the similarity probability measure on $\Sigma$ defined in
Section \ref{subsec:33}.

\begin{theorem}{\cite{GLMM14}}\label{th4.1n} For every similarity random subfractal $F_\sigma$ of $F$,  {the dimension spectrum}  $\mathrm{sp}(F_\sigma)$
 is an interval satisfying
$[s^*\frac{\log(k-1)-\log(r-1+ A (k-r))}{\log \frac{1}{a}},
s^*]\subseteq \mathrm{sp}(F_\sigma) \subseteq [s^*\frac{\log k -
\log r}{\log \frac{1}{A}}, s^*]$,
where $s^*=\sdim(S)$,
 $a=\min\{\pi_S(i)\,|\, i\in\Sigma\}$, and
  $A=\max\{\pi_S(i)\,|\, i\in\Sigma\}$.

 {In particular}, if all the
  contraction ratios of $F$ have the same value $c$, then every
  similarity-random (i.e., uniformly random) subfractal $F_\sigma$ of $F$
  has dimension spectrum
\begin{equation*}\mathrm{sp}(F_\sigma)= [s^*(1-\frac{
\log r}{\log k}), s^*], \end{equation*} where $s^*=\sdim(S)=(\log
k)/(\log \frac{1}{c})$. \end{theorem}

Many challenging open questions remain concerning the analysis of
the dimension of points in more general versions of random fractals,
both by completing the results in \cite{GLMM14}\ to random selectors
for different probability measures and by considering
generalizations such as self-affine fractals and fractals with
randomly chosen contraction ratios.

   \subsection{Beyond Euclidean Spaces}

While Euclidean space has a very well-behaved metric based on a
Borel measure $\mu$, where for instance $s$-Hausdorff measure
coincides with $\mu$ for $s=1$, this is not the case for other
metric spaces. Since both Hausdorff and packing dimension can be
defined in any metric space, the second author has considered in
\cite{edgms}\ the extension of algorithmic dimension to a large
class of separable metric spaces, the class of spaces with a
computable nice cover. This extension includes an algorithmic
information characterization of constructive dimension, based on the
concept of Kolmogorov complexity of a point at a certain precision,
which is an extension of the concept presented in section
\ref{sec:2}\ for Euclidean space.

   \subsection{Beyond Computability}

Resource-bounded dimension, introduced  in \cite{DCC}\ by the first
author, has been a very fruitful tool in the quantitative study of
complexity classes, see \cite{FGCC,Lut05}\ for the main results.
Many of the main complexity classes have a suitable resource bound
for which the corresponding dimension is adequate for the class,
since it has maximal value for the whole class.

The development of resource-bounded dimension was based on a
characterization of Hausdorff dimension in terms of betting
strategies, imposing different complexity constraints on those
strategies to obtain the different resource-bounded dimensions.
Contrary to the case of computability constraints introduced in
section \ref{sec:3}, many important resource-bounds such as
polynomial time dimension do not have corresponding algorithmic
information characterizations (although more elaborated compression
algorithms characterizations have been obtained in
\cite{DIC,Hitchcock:DERC}).

In fact the study of gambling under very low complexity constraints,
finite-state computability, has been studied at least since the
seventies \cite{SchSti72,Fede91}\ and the corresponding effective
dimension, finite-state dimension, was studied by Dai, Lathrop, and
the two authors \cite{FSD}\ where finite-state dimension is
characterized in terms of finite-state compression.

 For the definition of resource-bounded dimension, a class of languages
$\mathcal{C}$ is represented via characteristic sequences as a set
of infinite binary sequences $\mathcal{C}\subseteq\{0,1\}^{\infty}$.
Using binary representation each language can be seen as a real
number in $[0,1]$ and resource-bounded dimension as a tool in
Euclidean space. Resource-bounded dimension has a natural extension
$\sinf$ for other finite alphabets $\Sigma$ and the first question
is therefore whether the choice of alphabet is relevant for the
study of Euclidean space. A satisfactory answer is given in
\cite{HM13}\ where it is proven that polynomial-time dimension is
invariant under base change, that is, for every base $b$  and set
$X\subseteq\R$ the set of base-$b$-representations of all elements
in $X$ has a polynomial-time dimension independent of $b$.

Finite-state dimension is not closed under base change, but its
connections with number theory are deep. Borel introduced normal
numbers in \cite{Bore09}, defining a real number $\alpha$ to be
Borel normal in base $b$ if  for every finite sequence $w$ of
base-$b$ digits, the asymptotic, empirical frequency of $w$ in the
base-$b$ expansion of $\alpha$ is $b^{-|w|}$. There is a tight
relationship of Borel-normality and finite-state dimension, since a
real number is normal in base $b$ iff its base $b$ representation is
a finite-state dimension 1 sequence \cite{SchSti72,BoHiVi05}. It is
known \cite{Cass59, Schm60}\ that there are numbers that are normal
in one base but not in another, so the nonclosure under base change
property of finite-dimension is a corollary of these results.
Absolutely normal numbers are real numbers that are normal in every
base, so they correspond to real numbers whose base-$b$
representation has finite-dimension 1 for every base $b$, this
characterization has been used in very effective constructions of
absolutely normal numbers \cite{BHS13,cannnlt}. It is natural to ask
whether there are real numbers for which the finite-state dimension
of its base-$b$ representations is strictly between 0 and 1 and does
not depend on the base $b$.

\subsection{Beyond Fractals}

This chapter's primary focus is the role of algorithmic fractal
dimensions in fractal geometry. However, it should be noted that
fractal geometry is only a part of geometric measure theory, and
that algorithmic methods may shed light on many other aspects of
geometric measure theory.

Many questions in geometric measure theory involve rectifiability
\cite{Fede69}. The simplest case of this classical notion is the
rectifiability of curves. A {\sl curve\/} in $\R^n$ is a continuous
function $f:  [0,1]\to \R^n$. The {\sl length\/} of a curve $f$ is
\begin{equation*}\length(f)=\sup_{\vec{a}}\sum_{i=0}^{k-1}|f(a_{i+1})-f(a_i)|,\end{equation*}
where the supremum is taken over all {\sl dissections\/} $\vec{a}$
of $[0,1]$, i.e., all $\vec{a}=(a_0, \ldots, a_k)$ with $0=a_0<a_1<
\ldots < a_k=1$. Note that $\length(f)$ is the length of the actual
path traced by $f$, which may ``retrace'' parts of its range. (In
fact, there are computable curves $f$ for which every computable
curve $g$ with the same range {\sl must\/} do unboundedly many such
retracings \cite{GLM11}.) A curve $f$ is {\sl rectifiable\/} if
$\length(f)<\infty$.

Gu and the authors \cite{GLM06}\ posed the fanciful question, ``
Where can an infinitely small nanobot go?'' Intuitively, the nanobot
is the size of a Euclidean point, and its motion is algorithmic, so
its trajectory must be a curve $f: [0,1]\to \R^n$ that is computable
in the sense of computable analysis \cite{Wei00}. Moreover, the
nanobot's trajectory $f$ should be rectifiable. This last
assumption, aside from being intuitively reasonable, prevents the
question from being trivialized by space-filling curves
\cite{Saga94, CDM12}.

The above considerations translate our fanciful question about a
nanobot to the following mathematical question. {\sl Which points in
$\R^n$ ($n\ge 2$) lie on  rectifiable computable curves?} In honor
of an anonymous, poetic reviewer who called the set of all such
points ``the beaten path'', we write $\mathrm{BP}^{(n)}$ for the set
of all points in $\R^n$ that lie on rectifiable computable curves.
The objective of \cite{GLM06}\ was to characterize the elements of
$\mathrm{BP}^{(n)}$.

A few preliminary observations on the set $\mathrm{BP}^{(n)}$ are in
order here. Every computable point in $\R^n$ clearly lies in
$\mathrm{BP}^{(n)}$, so $\mathrm{BP}^{(n)}$ is a dense subset of
$\R^n$. It is also easy to see that $\mathrm{BP}^{(n)}$ is
path-connected. On the other hand, the ranges of rectifiable curves
have Hausdorff dimension 1 \cite{Falc14}\ and there are only
countably many computable curves, so $\mathrm{BP}^{(n)}$ is a
countable union of sets of Hausdorff dimension 1 and hence has
Hausdorff dimension 1. Since $n\ge 2$, this implies that most points
in $\R^n$ do not lie on the beaten path $\mathrm{BP}^{(n)}$.

For each rectifiable computable curve $f$, the set range $f$ is a
computably closed, i.e., $\Pi^0_1$, subset of $\Rn$. By the
preceding paragraph and Hitchcock's correspondence principle
(\ref{equ311}), it follows that $\cdim(\mathrm{BP}^{(n)})=1$, whence
every point $x\in \mathrm{BP}^{(n)}$ satisfies $\dim(x)\le 1$. This
is a necessary, but not sufficient condition for membership in
$\mathrm{BP}^{(n)}$, because the complement of $\mathrm{BP}^{(n)}$
contains points of arbitrarily low dimension \cite{GLM06}.
Characterizing membership in $\mathrm{BP}^{(n)}$ thus requires
algorithmic methods to be extended beyond fractal dimensions.

The ``analyst's traveling salesman theorem'' of geometric measure
theory characterizes those subsets of Euclidean space that are
contained in rectifiable curves. This celebrated theorem was proven
for the plane by Jones \cite{Jone90}\ and extended to
high-dimensional Euclidean spaces by Okikiolu \cite{Okik91}. The
main contribution of \cite{GLM06}\ is to formulate the notion of a
{\sl computable Jones constriction,} an algorithmic version of the
infinitary data structure implicit in the analyst's traveling
salesman theorem, and to prove the {\sl computable analyst's
traveling salesman theorem,} which says that a point in Euclidean
space lies on the beaten path $\mathrm{BP}^{(n)}$ if and only if it
is ``permitted'' by some computable Jones constriction.

The computable analysis of points in rectifiable curves has
continued in at least two different directions. In one direction,
Rettinger and Zheng have shown (answering a question in
\cite{GLM06}) that there are points in $\mathrm{BP}^{(n)}$ that do
not lie on any computable curve of computable length \cite{RZ09}\
and extended this to obtain a four-level hierarchy of simple
computable planar curves that are {\sl point-separable\/} in the
sense that the sets of points lying on curves of the four types are
distinct \cite{ZR12}. In another direction, McNicholl \cite{McN13a}\
proved that there is a point on a computable arc (a set computably
homeomorphic to $[0,1]$) that does not lie in $\mathrm{BP}^{(n)}$.
In the same paper, McNicholl used a beautiful geometric priority
argument to prove that there is a point on a computable curve of
computable length that does not lie on any computable arc.

It is apparent from the above results that algorithmic methods will
have  a great deal more to say about rectifiability and other
aspects of geometric measure theory.

\begin{acknowledgement}
The first author's work was supported in part by National Science
Foundation research grants 1247051 and 1545028 and is based in part
on lectures that he gave at the New Zealand Mathematical Research
Institute Summer School on Mathematical Logic and Computability,
January 9-14, 2017.  He thanks Neil Lutz for useful discussions and
Don Stull and Andrei Migunov for helpful comments on the exposition.
The second author's research was supported in part by Spanish
Government
   MEC Grants TIN2011-27479-C04-01 and TIN2016-80347-R and  was
  done and  based in part on lectures given during a research stay at the Institute for Mathematical
  Sciences at the National University of Singapore for the Program
  on Aspects of Computation, August 2017. We thank two anonymous
  reviewers for useful suggestions on the exposition.
\end{acknowledgement}




\bibliographystyle{spmpsci_cca}
\bibliography{Bib_LM,cca}
\end{document}